%% file: article.tex
\documentclass[11pt, letterpaper]{article}
\def\Anonymise{0}

\usepackage{geometry}
\usepackage{palatino}
\usepackage{inconsolata}

\usepackage[T1]{fontenc}
\usepackage{graphicx}
\usepackage{algorithm,algpseudocode}
\usepackage{xspace}
\usepackage{xcolor}
\usepackage[utf8]{inputenc}
\usepackage[noadjust]{cite}
\usepackage{amsmath, amssymb, amsfonts, amsthm}
\usepackage{enumerate}
\usepackage[linktocpage=true,pagebackref=true,colorlinks,linkcolor=magenta,citecolor=blue,bookmarks,bookmarksopen,bookmarksnumbered]{hyperref}
\usepackage{paralist}
\usepackage{gensymb}
\usepackage{rotating}
\usepackage{mdframed}
\usepackage{multirow}
\usepackage{verbatim} 
\usepackage[nottoc]{tocbibind}%
\usepackage{mleftright}
\usepackage[draft,inline,marginclue,index]{fixme}  %
\usepackage{tablefootnote}
\usepackage[capitalize]{cleveref}%

\newtheorem{theorem}{Theorem}[section]

\newtheorem{claim}[theorem]{Claim}

\theoremstyle{definition}

\newtheorem{problem}{Problem}

\theoremstyle{remark}

\renewcommand*\backref[1]{\ifx#1\relax \else (cit. on p. #1) \fi}

\algtext*{EndWhile}%
\algtext*{EndIf}%
\algtext*{EndFor}%
\algtext*{EndFunction}%

\algnewcommand{\IfThen}[2]%
{\State \algorithmicif\ #1\ \algorithmicthen\ #2}

\makeatletter
\def\moverlay{\mathpalette\mov@rlay}
\def\mov@rlay#1#2{\leavevmode\vtop{%
	\baselineskip\z@skip \lineskiplimit-\maxdimen
	\ialign{\hfil$\m@th#1##$\hfil\cr#2\crcr}}}
\newcommand{\charfusion}[3][\mathord]{
	#1{\ifx#1\mathop\vphantom{#2}\fi
		\mathpalette\mov@rlay{#2\cr#3}
	}
	\ifx#1\mathop\expandafter\displaylimits\fi}
\makeatother

\def\poly{\mathrm{poly}}
\def\polylog{\mathrm{polylog}}

\def\Pans{{P_{\sf ans}}}

\def\eps{\epsilon}
\def\rank{\mathrm{rank}}

\def\HitSet{\textsc{HitSet}}
\def\ellmax{\ell_{\sf max}}
\def\Helper{\textsc{Helper}}

\def\caD{\mathcal{D}}

\def\Dstar{\caD_{\star}}

\newcommand{\Patrascu}{P\v{a}tra\c{s}cu\xspace}
\newenvironment{mybox}
{\begin{mdframed}[hidealllines=true,backgroundcolor=gray!10,innertopmargin=0]}
{\end{mdframed}}

\fxsetup{theme=color,mode=multiuser}

\definecolor{color1}{RGB}{46,134,193}
\FXRegisterAuthor{rhl}{arhl}{\colorbox{color1}{\color{white}Hanlin}}

\graphicspath{{./figs/}}

\begin{document}
	
	\title{Maintaining Exact Distances under Multiple Edge Failures}
	
	\ifnum\Anonymise=1
	\author{Anonymous author(s)}
	\else
	\author{Ran Duan \\ \small{Tsinghua University} \\ \texttt{\small \href{mailto:duanran@mail.tsinghua.edu.cn}{duanran@mail.tsinghua.edu.cn}}
	\and
	Hanlin Ren\thanks{Most of this work was done when Hanlin Ren was affiliated with Tsinghua University.} \\ \small{University of Oxford} \\ \texttt{\small \href{mailto:hanlin.ren@cs.ox.ac.uk}{hanlin.ren@cs.ox.ac.uk}}}
	\fi

	\pagenumbering{gobble}
	
	\maketitle
	\input{abstract.tex}

	\tableofcontents
	\newpage
	\pagenumbering{arabic}

\input{intro.tex}

\input{overview.tex}

	\input{exact-oracle.tex}

\input{hitset.tex}

	\input{conclusion.tex}

	\section*{Acknowledgements}
	We thank Yong Gu and Tianyi Zhang for helpful discussions during the initial stage of this research. We are grateful to Yaowei Long and Lijie Chen for helpful comments on a draft version of this paper.
	
	\bibliography{article}	
	
	\newpage

\end{document}

%% file: abstract.tex
\begin{abstract}
	We present the first compact distance oracle that tolerates multiple failures and maintains \emph{exact} distances. Given an undirected weighted graph $G = (V, E)$ and an arbitrarily large constant $d$, we construct an oracle that given vertices $u, v \in V$ and a set of $d$ edge failures $D$, outputs the \emph{exact} distance between $u$ and $v$ in $G - D$ (that is, $G$ with edges in $D$ removed). Our oracle has space complexity $O(d n^4)$ and query time $d^{O(d)}$. Previously, there were compact \emph{approximate} distance oracles under multiple failures [Chechik, Cohen, Fiat, and Kaplan, SODA'17; Duan, Gu, and Ren, SODA'21], but the best exact distance oracles under $d$ failures require essentially $\Omega(n^d)$ space [Duan and Pettie, SODA'09]. Our distance oracle seems to require $n^{\Omega(d)}$ time to preprocess; we leave it as an open question to improve this preprocessing time.
\end{abstract}

%% file: intro.tex
\section{Introduction}
\label{sec:intro}

Real-life networks are dynamic. Sometimes a link or node suffers from a crash failure and has to be removed from the network. Occasionally a new link or node is added to the network. This motivates the study of \emph{dynamic} graph algorithms: algorithms that receive a stream of updates to the graph and needs to simultaneously respond to queries about the \emph{current} graph. The field of dynamic graph algorithms is both classical and vibrant that we will not be able to survey here.

However, in many situations, the network is not ``too'' dynamic in the sense that the graph will always remain close to a ``base'' graph. Thus, by preprocessing this base graph, one might obtain better query time bounds than what is possible in the fully dynamic setting. One example is the \emph{$d$-failure} model: in each query, there is a (small) set of failures (which are either vertices or edges), and we are interested in the graph with the failures removed. After this query was done, the failures are repaired and do not influence the next query.

In this paper, we consider the problem of maintaining distances in the $d$-failure model. More precisely, given a graph $G = (V, E)$, we want to build an oracle that answers the following queries quickly: given a set of edge failures $D\subseteq E$ with size at most $d$ and two vertices $u, v \in V$, what is the distance from $u$ to $v$ in $G-D$ (i.e., the graph with edges in $D$ removed)?

\subsection{Previous Works}

The case of $d = 1$ (i.e., only one failure) is well-understood. There is an oracle of size $O(n^2\log n)$ and query time $O(1)$ that maintains exact distances in a directed graph under one edge or vertex failure \cite{DTCR08}. A long line of work \cite{BK08, BK09, WY13, GW20, DZ17, ChechikC20, Ren22, GuR21} has focused on optimising the space or preprocessing time of this oracle.

The case of $d = 2$ was considered by Duan and Pettie \cite{DP09}: they presented an oracle of size $O(n^2\log^3 n)$ and query time $O(\log n)$ that maintains exact distances in directed graphs under two vertex or edge failures. Unfortunately, their techniques do not seem to generalise to even $3$ failures. They even concluded that ``moving beyond dual-failures will require a fundamentally different approach to the problem.''

The problem becomes significantly harder when $d$ becomes large. In fact, previous works in this regime had to weaken the query requirements: instead of exact distance queries, they could only handle connectivity queries or \emph{approximate} distance queries.
\begin{enumerate}[(1)]
	\item \Patrascu and Thorup \cite{PT07} presented an oracle for handling connectivity queries under $d$ edge failures in an undirected graph. Their oracle has size $\tilde{O}(m)$\footnote{$\tilde{O}$ hides $\polylog(n)$ factors.} and query time $\tilde{O}(d)$. Duan and Pettie \cite{DP10, DP20} presented oracles that answer connectivity queries under $d$ \emph{vertex} failures\footnote{Edge failures are always no harder than vertex failures, as we can always insert a vertex in the middle of every edge to simulate an edge failure by a vertex failure. However, in many cases \cite{DP10, DuanGR21}, dealing with vertex failures requires significantly new ideas compared to edge failures.}, which has size $\tilde{O}(m)$ and query time $\tilde{O}(d^2)$.
	\item Chechik, Langberg, Peleg, and Roditty \cite{CLPR12} designed $O(kd)$-approximate distance oracles under $d$ edge failures in an undirected graph, which has size $O(dkn^{1+1/k}\log (nW))$ and query time $\tilde{O}(d)$; here $k$ is an arbitrary constant and $W$ is an upper bound on the edge weights. Bil\`o, Gual\`a, Leucci, and Proietti \cite{BGLP16} improved the approximation ratio to $(2d+1)$, with the expense of a larger query time of $\tilde{O}(d^2)$.
	\item If we are willing to tolerate a space complexity bound exponential in $d$, then we could even achieve $(1+\eps)$-approximation for every $\eps > 0$. Chechik, Cohen, Fiat, and Kaplan \cite{CCFK17} designed $(1+\eps)$-approximate distance oracles under $d$ edge failures in undirected graphs, which has size $O(n^2(\log n / \eps)^d\cdot d\log W)$ and query time $O(d^5\log n\log\log W)$. Duan, Gu, and Ren \cite{DuanGR21} generalised this oracle to also handle vertex failures; their oracle has size $n^{2.01}\cdot (\log n/\eps)^{O(d)}\cdot \log W$ and query time $\poly(d, \log n, \log \log W)$.
	\item Brand and Saranurak \cite{BS19} designed reachability and exact distance oracles for \emph{directed} graphs. Their reachability oracle has size $\tilde{O}(n^2)$ and query time $O(d^\omega)$, while their distance oracle has size $\tilde{O}(Wn^{2+\alpha})$ and query time $O(Wn^{2-\alpha}d^2 + Wnd^\omega)$; here $\omega < 2.3728596$ is the matrix multiplication exponent \cite{CW90, Sto10, Wil12, LeGall, AlmanW21} and $\alpha \in [0, 1]$ is an arbitrary parameter.\label{item:previous-work-d}
\end{enumerate}

\paragraph{Exact distances?} Despite significant efforts on the $d$-failure model, our understanding about the setting where \emph{exact} distances need to be maintained is still quite primitive. One reason is that the structure of shortest paths after $d$ failures appears to be very complicated: \cite{DP09} used heavy case analysis to deal with \emph{two} failures, and three failures appear to be even harder! (We also think this complexity provides further motivation for studying exact distance oracles in the $d$-failure model, as a good oracle enhances our understanding of the structure of $d$-failure shortest paths.)

All oracles in the above list could only answer connectivity or approximate distance queries. The only exception is the distance oracle in (\ref{item:previous-work-d}), but its query time is polynomial in $n$ and $W$. Ideally, we want an oracle with query time $\poly(\log n, \log W)$. Thus the following question is open:

\begin{mybox}
\begin{problem}
	Fix a large constant $d$. Is there a $d$-failure oracle for handling exact distance queries in undirected graphs with query time $\poly(\log n, \log W)$ and a reasonable size bound?\label{open1}
\end{problem}
\end{mybox}

In fact, before our work, the best $d$-failure exact distance oracle requires size $\tilde{\Omega}(n^d)$ \cite{DP09}, only slightly better than the trivial $O(n^{d+2})$ bound.\footnote{For simplicity, this particular paragraph only considers vertex failures. The na\"ive bound for edge failures is $O(m^{d-2}n^2)$ which is even worse than the na\"ive bound for vertex failures.} The trivial bound is obtained as follows: for every set of failures $D$ with size at most $d$, we store the all-pairs shortest path matrix for the graph $G - D$, which requires $\binom{n}{d} \cdot O(n^2) \le O(n^{d+2})$ space complexity. Duan and Pettie \cite{DP09} observed that their dual-failure oracle helps shave a factor of about $n^2$ from this trivial bound: for every set of failures $D$ with size at most $d-2$, we build a dual-failure exact distance oracle for $G-D$ which occupies size $\tilde{O}(n^2)$, and the total space complexity becomes $\binom{n}{d-2}\cdot \tilde{O}(n^2) \le \tilde{O}(n^d)$. However, even the answer to the following problem was unknown:

\begin{mybox}
	\begin{problem}
		Is there a $100$-failure exact distance oracle for undirected graphs with query time $\poly(\log n, \log W)$ and size $O(n^{99.9})$?\label{open2}
	\end{problem}
\end{mybox}

\subsection{Our Results}
Our main result is an exact distance oracle under $d$ edge failures, for every constant $d\ge 1$.

\begin{mybox}
	\begin{theorem}
		Let $G = (V, E)$ be an undirected weighted graph. For every constant $d \ge 1$, there is an oracle that handles exact distance queries in $G$ under $d$ edge failures. The oracle has size $O(n^4)$ and query time $O(1)$.\label{thm: main}
	\end{theorem}
\end{mybox}

Our oracle is the first one that maintains \emph{exact} distances under multiple (say $100$) failures while having reasonable size and query time bounds ($O(n^4)$ and $O(1)$ respectively). In particular, we answer Problems~\ref{open1} and \ref{open2} affirmatively. Unfortunately, we do not know how to preprocess our oracle in less than $n^{\Omega(d)}$ time. We leave it as an open problem to improve the preprocessing time of our oracle.

Our oracle also extends to super-constant $d$. In this case, our oracle has query time $d^{O(d)}$ and size $O(dn^4)$. We note that an exponential dependence on $d$ also occurs in the best data structures for maintaining $(1+\eps)$-approximate distances \cite{CCFK17, DuanGR21} (albeit in the space complexity bounds instead of the query time bounds).

\subsection{Notation}
Let $D\subseteq E$, we use $G-D$ to denote the graph $G$ with edges in $D$ removed. For a graph $H$ and two vertices $u, v\in V(H)$, $\pi_H(u, v)$ denotes the shortest path from $u$ to $v$ in $H$, and $|\pi_H(u, v)|$ denotes the length of this shortest path. When $H = G$ is the input graph, we may omit the subscript $H$, i.e., $\pi(u, v) = \pi_G(u, v)$. Although this paper only considers undirected graphs, the paths will be directed, i.e., $\pi_H(u, v)$ is a path \emph{from $u$ to $v$} and $u$ (or $v$) is the first (or last) vertex on it.

We assume that the shortest paths in $G$ are unique. If not, we could randomly perturb the edge weights of $G$ by a small value; the correctness follows from the isolation lemma \cite{MulmuleyVV87, Ta-Shma15}. Alternatively, we could use a method described in \cite[Section 3.4]{DI04} to obtain unique shortest paths. We omit the details here.

We will also consider shortest path trees (in the input graph $G$). For vertices $v, v' \in V$, $T_v$ denotes the shortest path tree rooted at $v$, and $T_v(v')$ denotes the subtree of $T_v$ rooted at $v'$ (which contains all vertices $w$ such that $v'$ is on the path $\pi(v, w)$).

%% file: overview.tex
\section{An Overview of Our Oracle}\label{sec: overview}
In this section, we present a high-level overview of our exact distance oracle.

\paragraph{A recursive approach.} Our starting point is the following structural theorem for the shortest paths in a graph with $d$ edge failures \cite[Theorem 2]{ABKCM02}. (See also \cref{thm:edge-decomposable}.)

\begin{theorem}\label{thm: decomposable (informal version)}
	Let $G$ be a graph, $D$ be a set of $d$ edge failures. Any shortest path in $G-D$ can be decomposed into the concatenation of at most $d+1$ shortest paths in $G$ interleaved with at most $d$ edges.
\end{theorem}

Given a query $(u, v, D)$, suppose we want to find $\Pans = \pi_{G-D}(u, v)$. From \cref{thm: decomposable (informal version)}, $\Pans$ is divided into $d+1$ segments where each segment is a shortest path in $G$. Our strategy is to find an arbitrary vertex $w \in \Pans$ \emph{which is neither on the first nor on the last segment}, recursively find $\pi_{G-D}(u, w)$ and $\pi_{G-D}(w, v)$, and concatenate these two paths. If $\Pans$ can be decomposed into $k$ segments and $w$ is neither on the first nor on the last segment, then both $\pi_{G-D}(u, w)$ and $\pi_{G-D}(w, v)$ can be decomposed into at most $k-1$ segments. It follows that the recursion depth is at most $d+1$.

We do not know how to find a single vertex $w$, but we are able to find a ``hitting set'' consisting of $\poly(d)$ many vertices $w$ such that at least one $w$ sits on $\Pans$ and is neither on the first nor on the last segment of $\Pans$. Therefore, the query time is $\poly(d)^d = d^{O(d)}$.

\paragraph{Finding a hitting set.} Now, it remains to design a procedure for finding such a hitting set. It suffices to find a set $H\subseteq V$ such that:
\begin{enumerate}[(I)]
	\item there is a vertex $w\in H$ that lies on $\Pans$;\label{item: overview -> hitting set}
	\item $|H|$ should be small; and\label{item: overview -> small}
	\item every vertex $w\in H$ that lies on $\Pans$ does not lie on the first or the last segment of $\Pans$.\label{item: overview -> the hitting set hits the middle of Pans}
\end{enumerate}

As a warm-up, we present a (very simple!) hitting set satisfying (\ref{item: overview -> hitting set}) and (\ref{item: overview -> small}). Of course it is (\ref{item: overview -> the hitting set hits the middle of Pans}) that enables us to upper bound the recursion depth; we will address (\ref{item: overview -> the hitting set hits the middle of Pans}) later.

For every $u, v\in V$, we simply let $\caD[u, v]$ be the set of $d$ edge failures \emph{whose removal maximises the distance from $u$ to $v$}. Note that $\caD[u, v]$ does not depend on $D$ and therefore can be preprocessed in advance. Since $|\caD[u, v]|\le d$, (\ref{item: overview -> small}) is true. Now we claim that either $\caD[u, v]$ hits $\Pans$, or $\pi_{G-\caD[u, v]}(u, v)$ and $\Pans$ are exactly the same path. Indeed, if $\caD[u, v]$ does not hit $\Pans$, then $\Pans$ is no shorter than the path $\pi_{G-\caD[u, v]}(u, v)$, which means $D$ should be \emph{as good as} the best candidate for $\caD[u, v]$! And $\pi_{G-\caD[u, v]}(u, v)$ and $\Pans$ should coincide in this case.

We regard the latter case (i.e.~$\pi_{G-\caD[u, v]}(u, v) = \Pans$) as ``trivial'' since it can be handled in preprocessing. Therefore, in this informal overview, we will ignore this case and simply say that $\caD[u, v]$ is a hitting set of $\Pans$.

Unfortunately, we are not aware of any fast algorithm for computing (any reasonable approximation of) $\caD[u, v]$, therefore we are currently unable to obtain a fast preprocessing algorithm for our data structure.

\paragraph{Achieving (\ref{item: overview -> the hitting set hits the middle of Pans}) via cleanness.} Suppose we have found a hitting set $H$ satisfying (\ref{item: overview -> hitting set}) and (\ref{item: overview -> small}) but not (\ref{item: overview -> the hitting set hits the middle of Pans}). Without loss of generality, suppose that there is a vertex $w\in H$ on the first segment of $\Pans$. This implies that $\pi(u, w)$ is intact from failures (as it lies entirely inside the first segment of $\Pans$). Our first insight is that if $w$ is \emph{``$u$-clean''}, then we could find another hitting set satisfying (\ref{item: overview -> hitting set}) and (\ref{item: overview -> small}) and \emph{avoiding the first segment of $\Pans$}.

The definition of cleanness is as follows: We say $w$ is \emph{$u$-clean} if both $\pi(u, w)$ and $T_u(w)$ are intact from failures.\footnote{To be more precise, ``$T_u(w)$ is intact from failures'' means that \emph{every vertex in $T_u(w)$ is not incident to any failed edge}.} (Recall that $T_u$ is the shortest path tree rooted at $u$, and $T_u(w)$ is the subtree of $T_u$ with root $w$.) Now suppose $w$ is $u$-clean and $\Pans$ goes through $w$. Let $\Dstar$ be the maximiser of $|\pi_{G-D'}(u, v)|$ such that 
\begin{equation}
\textit{both $\pi(u, w)$ and $T_u(w)$ are intact from $D'$.}
\label{eq: overview -> condition on caD}
\tag{$\sigma$}
\end{equation}
Again, $\Dstar$ depends on $u, v, w$ but not on $D$, so it can be preprocessed in advance. As $w$ is $u$-clean, $D$ also satisfies (\ref{eq: overview -> condition on caD}). Therefore $\Dstar$ hits $\Pans$ by the above reasoning.

Now consider a vertex $w'$ incident to some edge in $\Dstar$ and suppose that $\Pans$ also goes through $w'$. If $\pi(u, w')$ is intact from $D$, then either $\pi(u, w')$ does not go through $w$ (in this case $\Pans$ does not go through $w$ either) or $w'$ is in $T_u(w)$ (violating (\ref{eq: overview -> condition on caD})), a contradiction. Therefore we only need to consider vertices $w' \in \Dstar$ such that $\pi(u, w')$ is not intact from $D$; such vertices $w'$ can never hit the first segment of $\Pans$.

One can similarly define the notion of $v$-cleanness (which we omit here). The above discussion generalises to the following statement: If we know two vertices $u'$ and $v'$ such that $u'$ is $u$-clean, $v'$ is $v$-clean, and $\Pans$ goes through both $u'$ and $v'$, then we can find a hitting set satisfying all three conditions above. Note that we need to store a set $\Dstar$ for each possible $(u, v, u', v')$, therefore our oracle requires $O(dn^4)$ space complexity.

\paragraph{Finding a clean vertex.} Now, the problem reduces to finding $u$-clean and $v$-clean vertices that hit $\Pans$. For simplicity, in this overview, we consider the scenario where we already know a $v$-clean vertex $v'$ that is on $\Pans$, and want to find a $u$-clean vertex $u'$ that also hits $\Pans$. We believe this scenario already captures our core technical ideas.

Consider the following na\"ive attempt. Suppose we have a vertex $u'$ but $T_u(u')$ contains some failures. Our goal is to ``push'' $u'$ to a deeper vertex in $T_u$ so that eventually $T_u(u')$ will be intact from $D$. Let $\Dstar$ be the maximiser of $|\pi_{G-D'}(u, v)|$ such that
\begin{equation}
\textit{all of $\pi(v', v)$, $T_v(v')$, and $\pi(u, u')$ are intact from $D'$.}
\label{eq: overview -> condition of caD II naive}
\tag{$\tau_{\textsf{na\"ive}}$}
\end{equation}

Consider any vertex $w$ incident to some edge in $\Dstar$. It turns out that if $w$ is not a \emph{strict} descendant of $u'$ (i.e.~$w\not\in T_u(u')$ or $w = u'$) then we can deal with $w$ easily. If $w$ is a strict descendant of $u'$, then we assign $u'\gets w$ (i.e.~push $u'$ down to $w$) and repeat. We have made some progress as we pushed $u'$ to $w$ and $T_u(w)$ is a subset $T_u(u')$. But after how many steps would $T_u(w)$ become intact from $D$? It seems that we might need $\Omega(n)$ steps before we push $u'$ down to some vertex $w$ which is $u$-clean.

\def\uLCA{u_{\sf LCA}}
Our second insight is the following modification to (\ref{eq: overview -> condition of caD II naive}). Let $\uLCA$ be the least common ancestor of all failures in $T_u(u')$. If $D$ is intact from $\pi(u, u')$, then \emph{$D$ should also be intact from $\pi(u, \uLCA)$}. Now, let $\Dstar$ be the maximiser of $|\pi_{G-D'}(u, v)|$ such that

\begin{equation}
\textit{all of $\pi(v', v)$, $T_v(v')$, and $\pi(u, \underline{\uLCA})$ are intact from $D'$.}
\label{eq: overview -> condition of caD II}
\tag{$\tau$}
\end{equation}

Note that $D$ still satisfies (\ref{eq: overview -> condition of caD II}) which means $\Dstar$ is still a valid hitting set. We can push $u'$ down to some vertex $w$ which is a strict descendant of $\uLCA$. Now comes the crucial observation: the number of failures in $T_u(w)$ is \emph{strictly smaller} than the number of failures in $T_u(u')$. Since there are at most $d$ failures in $T_u(u')$, it takes at most $d$ steps of ``pushing $u'$ down'' before $u'$ becomes $u$-clean (i.e.~$T_u(u')$ becomes intact from $D$).

We remark that in the formal proof in \cref{sec: case II} there is no implementation of ``pushing $u'$ down''. Instead, we enumerate the portion of $T_u$ where the last step of pushing happens (i.e., $u'$ ``becomes'' $u$-clean). Nevertheless, the two formulations are equivalent, and we find the above description more intuitive.

%% file: exact-oracle.tex
\section{An Exact Distance Oracle for Edge Failures}\label{sec:exact-oracle}
In this section, we describe the framework for our exact distance oracle that handles $d$ edge failures and prove \cref{thm: main}.

As mentioned in \cref{sec: overview}, we use the structural theorem of shortest paths under edge failures in \cite{ABKCM02}. We say that a path is a \emph{$k$-decomposable path} if it is the concatenation of at most $k+1$ shortest paths in $G$, interleaved with at most $k$ edges. We have:
\begin{theorem}[Theorem 2 of \cite{ABKCM02}]
	For any set $D$ of $d$ edge failures in the graph and any vertices $u,v\in V$, the path $\pi_{G-D}(u, v)$ is a $d$-decomposable path.
	\label{thm:edge-decomposable}
\end{theorem}

For a set $D$ of $d$ edge failures and vertices $u,v\in V$, we define $\rank_{G-D}(u, v)$ as the smallest number $i$ such that $\pi_{G-D}(u, v)$ is an $i$-decomposable path. For example, $\rank_{G-D}(u, v) = 0$ if and only if the shortest $u$-$v$ path contains no failures. Thus, the conclusion of \cref{thm:edge-decomposable} can be interpreted as ``$\rank_{G-D}(u, v) \le |D|$.''

Our query algorithm relies on a subroutine $\HitSet(u, v, D)$, whose details will be given in \cref{sec:hitset}. Given $u,v\in V$ and a set of $d$ edge failures $D$, the output of $\HitSet(u, v, D)$ consists of an upper bound $L$ of $|\pi_{G-D}(u, v)|$ and a set of vertices $H\subseteq V$, such that the following holds. \begin{enumerate}[(a)]
	\item Either $|\pi_{G-D}(u, v)| = L$, or $\pi_{G-D}(u, v)$ goes through some vertex in $H$.\label{item1: HitSet hits the path}
	\item $|H| \le O(d^6)$.
	\item For every vertex $w \in H$, both $\pi(u, w)$ and $\pi(w, v)$ contain failures in $D$.\label{item1: HitSet doesn't hit the first or last segment}
\end{enumerate}

Note that Item (\ref{item1: HitSet doesn't hit the first or last segment}) ensures the following property:
\begin{claim}\label{claim: recursion on w}
	Let $u, v\in V$, $D$ be a set of $d$ edge failures, and $r = \rank_{G-D}(u, v)$. Suppose $w$ is a vertex in $H$ that is also on $\pi_{G-D}(u, v)$. Then $\rank_{G-D}(u, w) \le r-1$ and $\rank_{G-D}(w, v) \le r-1$.
\end{claim}
\begin{figure}[H]
	\centering
	\includegraphics[width=0.9\linewidth]{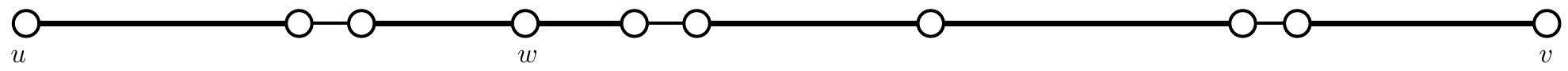}
	\caption{An example of \cref{claim: recursion on w}. Here $\rank_{G-D}(u, v) = 4$ and $\pi_{G-D}(u, v)$ is decomposed into $5$ shortest paths in $G$ (depicted as bold segments) interleaved with $3$ ($\le \rank_{G-D}(u, v)$) edges.}
\end{figure}
\begin{proof}
	We decompose $\pi_{G-D}(u, v)$ into $r+1$ shortest paths interleaved with at most $r$ edges. It is easy to see that $w$ is neither on the first shortest path nor on the last shortest path, as otherwise either $\pi(u, w)$ or $\pi(w, v)$ is intact from $D$, contradicting our assumption that $w \in H$.
	
	Since $w$ is not on the last shortest path, the portion of the decomposition from $u$ to $w$ includes at most $r$ shortest paths interleaved with at most $r-1$ edges, therefore $\rank_{G-D}(u, w) \le r-1$. Similarly, since $w$ is not on the first shortest path, $\rank_{G-D}(w, v) \le r-1$.
\end{proof}

We illustrate the query algorithm in \cref{alg:exact-query}. Roughly speaking, the idea is to enumerate a hitting vertex $w \in H$ and recursively find $\pi_{G-D}(u, w)$ and $\pi_{G-D}(w, v)$.

\begin{algorithm}[H]
	\caption{Query Algorithm for the Exact Distance Oracle}
	\label{alg:exact-query}
	\begin{algorithmic}[1]
		\Function{Query-r}{$u,v,D,r$}\Comment{We assume that $\rank_{G-D}(u, v) \le r$.}
		\IfThen {$\pi(u,v)\cap D=\varnothing$}{\Return $|\pi(u,v)|$}
		\IfThen {$r=0$}{\Return $+\infty$}
		\State {$(L, H)\gets\Call{HitSet}{u,v,D}$}
		\State {$ans\gets L$}
		\For {each $w$ in $H$}
			\State {$ans\gets \min\{ans, \Call{Query-r}{u,w,D,r-1} + \Call{Query-r}{w,v,D,r-1}\}$}
		\EndFor
		\State\Return $ans$
		\EndFunction
		\Function{Query}{$u,v,D$}
		\State\Return {$\Call{Query-r}{u,v,D,|D|}$}
		\EndFunction
	\end{algorithmic}
\end{algorithm}

In \cref{sec:hitset}, we will show that an invocation of $\HitSet(u, v, D)$ takes $\poly(d)$ time. Therefore, our query algorithm runs in $O(d^6)^d\cdot \poly(d) \le d^{O(d)}$ time. The following theorem demonstrates the correctness of our query algorithm.

\begin{theorem}
	Assuming the correctness of the $\HitSet$ structure, the query algorithm is correct.
\end{theorem}
\begin{proof}
	We show that for every $u, v, D, r$, if $\rank_{G-D}(u, v) \le r$, then $\textsc{Query-r}(u, v, D, r) = |\pi_{G-D}(u, v)|$. The theorem follows from \cref{thm:edge-decomposable}.
	
	Actually, it is easy to see that $\textsc{Query-r}(u, v, D, r) \ge |\pi_{G-D}(u, v)|$ as we could always construct a $u$-$v$ path in $G-D$ with length $\textsc{Query-r}(u, v, D, r)$. Therefore it suffices to show that $\textsc{Query-r}(u, v, D, r) \le |\pi_{G-D}(u, v)|$.
	
	We use induction on $r$. Our assertion is clearly true for $r = 0$. Let $r \ge 1$ and $(L, H) = \HitSet(u, v, D)$. If $|\pi_{G-D}(u, v)| = L$ then we are done, as $\textsc{Query-r}(u, v, D, r) \le L$ in this case. Otherwise by Item~(\ref{item1: HitSet hits the path}) in the correctness of $\HitSet$, $\pi_{G-D}(u, v)$ hits some vertex $w \in H$. By \cref{claim: recursion on w}, $\rank_{G-D}(u, w) \le r-1$ and $\rank_{G-D}(w, v) \le r-1$. By induction, we have that
	\[\textsc{Query-r}(u, w, D, r - 1) = |\pi_{G-D}(u, w)|\text{ and }\textsc{Query-r}(w, v, D, r - 1) = |\pi_{G-D}(w, v)|.\]
	It follows that
	\[\textsc{Query-r}(u, v, D, r) \le |\pi_{G-D}(u, w)| + |\pi_{G-D}(w, v)| = |\pi_{G-D}(u, v)|.\qedhere\]
\end{proof}

%% file: hitset.tex
\section{The $\HitSet$ Structure}\label{sec:hitset}

In this section, we describe the implementation of $\HitSet$. Recall that given $u, v\in V$ and a set of $d$ edge failures $D$, it should output an upper bound $L$ of $|\pi_{G-D}(u, v)|$ and a set of vertices $H\subseteq V$, such that the following items hold.
\begin{enumerate}[(a)]
	\item Either $|\pi_{G-D}(u, v)| = L$, or $\pi_{G-D}(u, v)$ goes through some vertex in $H$.\label{item: HitSet hits the path}
	\item $|H| \le O(d^6)$. \label{item: HitSet is small}
	\item For every vertex $w \in H$, both $\pi(u, w)$ and $\pi(w, v)$ contain failures in $D$.\label{item: HitSet doesn't hit the first or last segment}
\end{enumerate}

\paragraph{Warm-up.} Suppose that we drop Item (\ref{item: HitSet doesn't hit the first or last segment}) above, then it is easy to present a data structure for $\HitSet$ with space complexity $O(dn^2)$. For every two vertices $u, v\in V$, let $\caD[u, v]$ be the set of $d$ edge failures that maximises $|\pi_{G - \caD[u, v]}(u, v)|$. The data structure simply stores $\caD[u, v]$ and $\ellmax(u, v) = |\pi_{G - \caD[u, v]}(u, v)|$. In each query $\HitSet(u, v, D)$, for every edge $e \in \caD[u, v]$, we arbitrarily pick an endpoint of $e$ and add it into $H$. Then we return $L = \ellmax(u, v)$ and $H$. Item (\ref{item: HitSet is small}) holds since $|\caD[u, v]| \le d$. Therefore it suffices to prove that Item (\ref{item: HitSet hits the path}) holds:

\begin{claim}\label{claim: warm-up}
	For every set of $d$ edge failures $D$, either $|\pi_{G - D}(u, v)| = \ellmax(u, v)$, or $\pi_{G - D}(u, v)$ goes through some vertex in $H$.
\end{claim}
\begin{proof}
	Actually, we can show that either $|\pi_{G-D}(u, v)| = \ellmax(u, v)$ or $\pi_{G-D}(u, v)$ goes through some edge in $\caD[u, v]$. Suppose that $\pi_{G - D}(u, v)$ does not intersect $\caD[u, v]$. That is, $\pi_{G - D}(u, v)$ is a valid path from $u$ to $v$ that does not go through $\caD[u, v]$. It follows that 
	\[|\pi_{G - D}(u, v)| \ge |\pi_{G - \caD[u, v]}(u, v)| = \ellmax(u, v).\]
	However, we also have $|\pi_{G - \caD[u, v]}(u, v)| \ge |\pi_{G - D}(u, v)|$ by the definition of $\caD[u, v]$. Therefore 
	\[|\pi_{G - D}(u, v)| = \ellmax(u, v).\qedhere\]
\end{proof}

The warm-up case demonstrates that it is easy to satisfy Items (\ref{item: HitSet hits the path}) and (\ref{item: HitSet is small}). All technical complications introduced in the rest of this section deals with Item (\ref{item: HitSet doesn't hit the first or last segment}).

\subsection{The Data Structure} \label{sec: data structure for HitSet}

For every four vertices $u, v, u', v' \in V$ and two Boolean variables $b_1, b_2 \in \{0, 1\}$, our data structure contains a size-$d$ edge set $\caD[u, v, u', v', b_1, b_2]$. This set corresponds to the scenario where we want to find $\pi_{G - D}(u, v)$ (where $D$ is a set of failures given in the query), and we know two intermediate vertices $u', v'$ satisfying the following properties:\begin{enumerate}[(i)]
	\item The paths $\pi(u, u')$ and $\pi(v', v)$ are intact from $D$.\label{item: pi(u, u') and pi(v', v) are intact from D}
	\item We \emph{assume} that $\pi_{G - D}(u, v)$ goes through both $u'$ and $v'$; in other words, $\pi(u, u')$ is a prefix of $\pi_{G - D}(u, v)$, and $\pi(v', v)$ is a suffix of $\pi_{G - D}(u, v)$.
\end{enumerate}

(An intuitive interpretation of $u'$ and $v'$ is as follows. We are trying to find a hitting vertex $w$ such that both $\pi(u, w)$ and $\pi(w, v)$ intersect $D$, so we can add $w$ into our hitting set $H$; $u'$ and $v'$ represent our failed attempts, i.e., hitting vertices $w$ where $\pi(u, w)$ or $\pi(w, v)$ did not happen to intersect $D$.)

The meaning of Boolean variables $b_1$ and $b_2$ are as follows. If $b_1 = 1$, then we require that $T_u(u')\cap V(D) = \varnothing$, where $V(D)$ is the set of vertices incident to some failure in $D$. (Recall that $T_u$ is the shortest path tree rooted at $u$, and $T_u(u')$ is the subtree of $T_u$ rooted at $u'$.) If $b_1 = 0$, we do not impose any condition on $T_u(u') \cap V(D)$. Similarly, if $b_2 = 1$ then $T_v(v')\cap V(D) = \varnothing$, while if $b_2 = 0$ then we do not impose any condition on $T_v(v')\cap V(D)$.

Naturally, we define $\caD[u, v, u', v', b_1, b_2]$ as the set $D'$ that maximises $|\pi_{G-D'}(u, v)|$, subject to Item (\ref{item: pi(u, u') and pi(v', v) are intact from D}) and the conditions imposed by $b_1$ and $b_2$. For example, $\caD[u, v, u', v', 0, 1]$ is the maximiser of $|\pi_{G-D'}(u, v)|$ among all size-$d$ edge sets $D'$ where (See also \cref{fig: illustration of D[u v u' v' 0 1]})%
\[\pi(u, u') \cap D' = \varnothing, \pi(v', v) \cap D' = \varnothing, \text{ and } T_v(v')\cap V(D') = \varnothing.\]

\begin{figure}[H]
	\centering
	\includegraphics{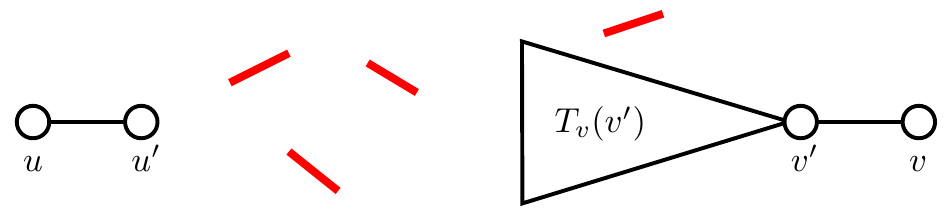}
	\caption{An illustration of $\caD[u, v, u', v', 0, 1]$. The red bold edges are a possible set of edge failures $D'$ that does not intersect $\pi(u, u')$, $\pi(v', v)$, and $T_v(v')$.}
	\label{fig: illustration of D[u v u' v' 0 1]}
\end{figure}

Our data structure occupies $O(dn^4)$ space. It is easy to see that our data structure can be preprocessed in time $m^{d+O(1)}$. We leave it as an open problem to improve this preprocessing time.

\subsection{The $\HitSet$ Algorithm}
In what follows, we say that a vertex $w \in V$ is \emph{$u$-clean} if $\pi(u, w) \cap D = \varnothing$ and $T_u(w) \cap V(D) = \varnothing$. Similarly, we say that a vertex $w \in V$ is \emph{$v$-clean} if $\pi(w, v) \cap D = \varnothing$ and $T_v(w) \cap V(D) = \varnothing$.

We present $\HitSet$ from special cases to the most general case.\begin{itemize}
	\item In Case I, we assume that we know two ``helper'' vertices $u', v'$ such that $u'$ is $u$-clean, $v'$ is $v$-clean, and $\pi_{G-D}(u, v)$ goes through both $u'$ and $v'$. This case is similar to the warm-up case and is essentially one invocation of the data structure $\caD[\cdot]$.
	\item In Case II, we assume that we know \emph{one} ``helper'' vertex $v'$ (such that $v'$ is $v$-clean and lies on $\pi_{G-D}(u, v)$), but we need to find the other ``helper'' vertex ($u'$). We will find a small number of candidate vertices $u'$, thus reducing this case to Case I. \emph{Most of our new techniques will appear in this case.}
	\item Case III is the most general case where we do not have any ``helper'' vertices and need to find them on our own. Nevertheless, using techniques similar to Case II, we can still find a small number of such ``helper'' vertices and reduce this case to Case II.
\end{itemize}

\subsubsection{Case I}
In this case, we assume that we already know vertices $u', v'\in V$ such that $u'$ is $u$-clean, $v'$ is $v$-clean, and $\pi_{G-D}(u, v)$ goes through both $u'$ and $v'$. (See \cref{fig: illustration of case I}.)

\begin{figure}[H]
	\centering
	\includegraphics{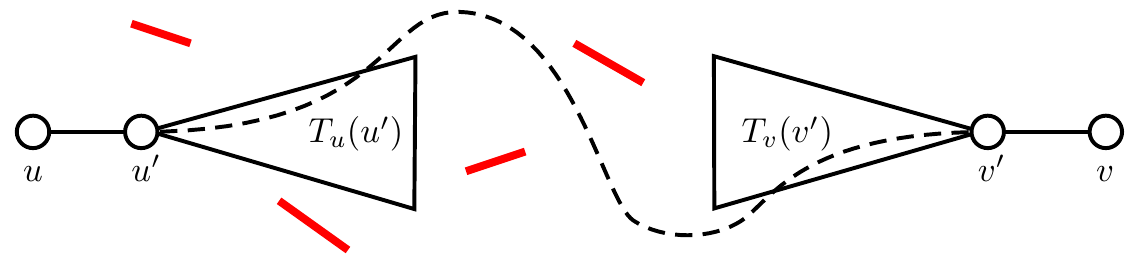}
	\caption{An illustration of Case I. Here, the red bold edges denote $D$. As $\pi(u, u')\cap D = \varnothing$ and $T_u(u')\cap V(D) = \varnothing$, $u'$ is $u$-clean. Similarly, $v'$ is $v$-clean. The dashed path denotes $\pi_{G-D}(u, v)$, which goes through both $u'$ and $v'$.}
	\label{fig: illustration of case I}
\end{figure}

This case can be solved similarly as in the warm-up case. Let $\Dstar = \caD[u, v, u', v', 1, 1]$, then $\Dstar$ is the maximiser of $|\pi_{G-D'}(u, v)|$ over all size-$d$ edge sets $D'$ such that
\begin{equation}
	\pi(u, u'), \pi(v', v), T_u(u'), \text{and }T_v(v')\text{ are intact from }D'.\tag{$\alpha$}\label{eq: requirement}
\end{equation}

Let $L = |\pi_{G-\Dstar}(u, v)|$ and
\[H = \{w \in V(\Dstar): \text{both $\pi(u, w)$ and $\pi(w, v)$ contain failures in $D$}\}.\]

It is easy to see that Items (\ref{item: HitSet is small}) and (\ref{item: HitSet doesn't hit the first or last segment}) hold, so it suffices to show Item (\ref{item: HitSet hits the path}), i.e.:
\begin{claim}\label{claim: case I}
	Either $|\pi_{G-D}(u, v)| = L$ or $\pi_{G-D}(u, v)$ goes through some vertex in $H$.
\end{claim}
\begin{proof}
	We first show that if $\pi_{G-D}(u, v)$ goes through some edge in $\Dstar$, then it also goes through some vertex in $H$. Suppose that $\pi_{G-D}(u, v)$ goes through some edge $e = (x, y) \in \Dstar$, we claim that $\pi(u, x)$ is not intact from $D$. Since $\pi_{G-D}(u, v)$ goes through $u'$ and $\pi(u, u')$ is intact from $D$, $\pi(u, u')$ coincides with the path from $u$ to $u'$ in $T_u$. Suppose $\pi(u, x)$ is also intact from $D$, then $x$ has to be either an ancestor or a descendant of $u'$ in $T_u$. Since $T_u(u') \cap V(\Dstar) = \varnothing$, $x$ cannot be a descendant of $u'$. Therefore, $x$ is a strict ancestor of $u'$. However, as $e$ is an incident edge of $x$ in the path $\pi_{G-D}(u, v)$, it has to be on the path $\pi(u, u')$, which contradicts the fact that $\pi(u, u')\cap \Dstar = \varnothing$. (See \cref{fig: illustration of case I claim}.)
	
	\begin{figure}[H]
		\centering
		\includegraphics{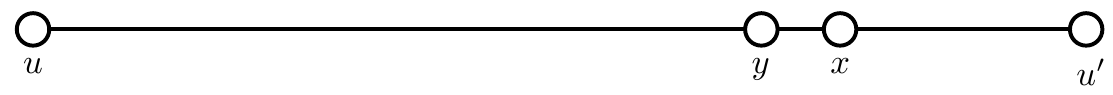}
		\caption{Illustration of \cref{claim: case I}.}
		\label{fig: illustration of case I claim}
	\end{figure}
	
	Therefore, $\pi(u, x)$ cannot be intact from $D$. Similarly, $\pi(x, v)$ cannot be intact from $D$ either. It follows that $x \in H$.
	
	Now the argument is essentially the same as \cref{claim: warm-up}. Suppose that $\pi_{G-D}(u, v)$ does not go through any edge in $\Dstar$, then
	\[|\pi_{G-D}(u, v)| \ge |\pi_{G-\Dstar}(u, v)| = L.\]
	However, $\Dstar$ is the maximiser of $|\pi_{G-D'}(u, v)|$ over all size-$d$ edge sets $D'$ satisfying the condition (\ref{eq: requirement}). Since $u'$ is $u$-clean and $v'$ is $v$-clean, $D$ also satisfies (\ref{eq: requirement}), thus
	\[|\pi_{G-D}(u, v)| = |\pi_{G-\Dstar}(u, v)| = L.\qedhere\]
\end{proof}

\subsubsection{Case II}\label{sec: case II}
In this case, we assume that we know a vertex $v' \in V$ such that $v'$ is $v$-clean, and $\pi_{G-D}(u, v)$ goes through $v'$. The goal of this case is to find a small number of candidates $u'$, such that every $u'$ is $u$-clean and $\pi_{G-D}(u, v)$ goes through one of these $u'$. In this way, we can reduce this case to Case I.

\def\LCA{\mathsf{LCA}}
\def\Tinduced{T_{\sf induced}}
\def\uroot{u_{\mathsf{root}}}
\def\Key{\mathsf{Key}}
\def\Tkey{T_{\sf key}}
First, we add a new vertex $\uroot$, add an edge $(\uroot, u)$ to connect it to $T_u$, and make $\uroot$ the root of $T_u$. This step is solely for convenience.

Denote $\Tinduced$ the induced subtree of $V(D)\cup \{\uroot\}$ over $T_u$, i.e.~an edge is in $\Tinduced$ if it is on some path between two vertices in $V(D)\cup \{\uroot\}$. We say a vertex $v$ is a \emph{key} vertex if either $v \in V(D)$ or the degree of $v$ in $\Tinduced$ is at least $3$. Let $\Key$ be the set of key vertices, then $|\Key| \le O(d)$. By contracting every non-key vertex in $\Tinduced$ (note that these vertices have degree exactly $2$), we obtain a smaller tree $\Tkey$ over $\Key$ where each edge $(x, y)$ in $\Tkey$ corresponds to a path from $x$ to $y$ in $\Tinduced$; here $x$ and $y$ are key vertices and all the intermediate vertices on the path have degree $2$.

\begin{figure}[H]
	\centering
	\includegraphics[width=0.9\linewidth]{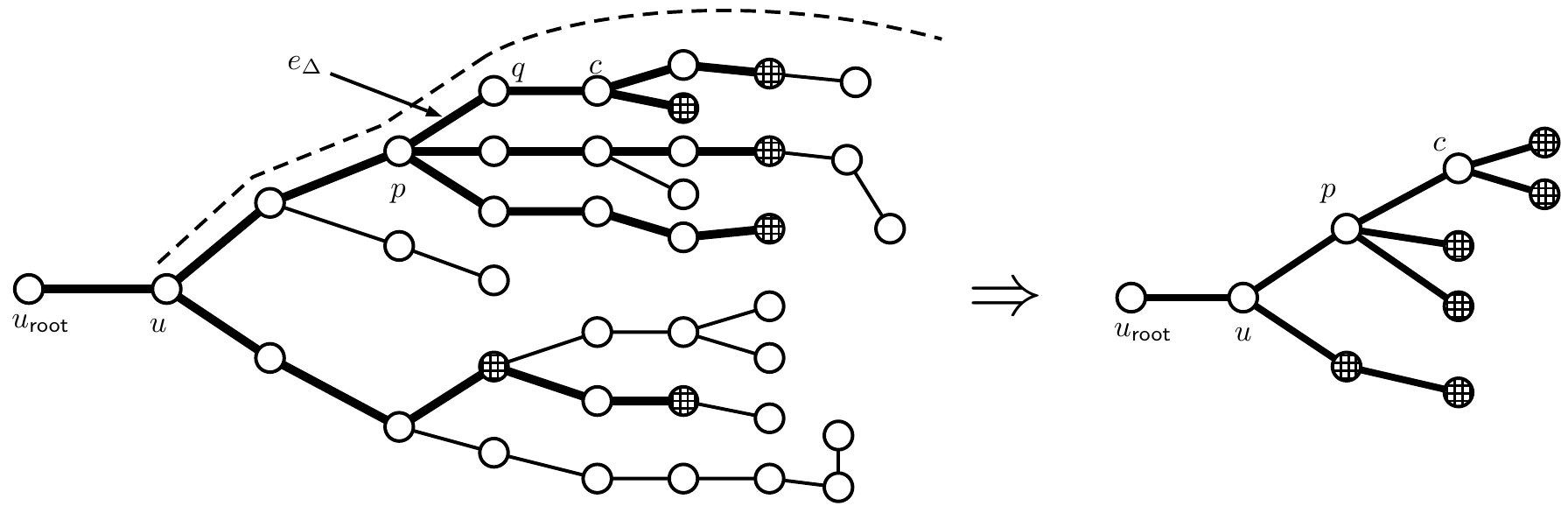}
	\caption{Left: a possible shortest path tree $T_u$. Hatched vertices are vertices in $V(D)$, and bold edges are edges in $\Tinduced$. The dash curve corresponds to $\pi_{G-D}(u, v)$, and $e_\Delta$ is the edge between $p$ and $q$. Right: the corresponding $\Tkey$. Note that $(p, c)$ is the edge in $\Tkey$ such that $e_\Delta$ is on the path from $p$ to $c$ in $\Tinduced$.}
\end{figure}

Let $e_\Delta$ be the last edge on $\pi_{G-D}(u, v)$ such that the portion from $u$ to $e_\Delta$ in $\pi_{G-D}(u, v)$ is entirely in $\Tinduced$. Note that $e_\Delta$ always exists since we added the auxiliary $\uroot$. (In particular, if $\pi_{G-D}(u, v)$ does not intersect $\Tinduced$ at all, we assume $e_\Delta$ is the edge between $\uroot$ and $u$.)

We enumerate an edge $(p, c) \in E(\Tkey)$ with the hope that $e_\Delta$ is on the path from $p$ to $c$ in $\Tinduced$. Note that there are $O(d)$ possible choices of $(p, c)$. Let $\Dstar = \caD[u, v, c, v', 0, 1]$, then $\Dstar$ maximises $|\pi_{G-\Dstar}(u, v)|$ over all size-$d$ sets of edge failures such that 
\begin{equation}
\Dstar \cap \pi(u, c) = \varnothing, \Dstar \cap \pi(v', v) = \varnothing, \text{ and }V(\Dstar) \cap T_v(v') = \varnothing.
\label{eq: condition for Dstar in Case II}\tag{$\beta$}
\end{equation}

If $D\cap\pi(u, c) \ne \varnothing$, then we discard the edge $(p, c)$. This is because the following claim shows that $e_\Delta$ cannot appear in the path from $p$ to $c$:

\begin{claim}
	If $D\cap\pi(u, c) \ne \varnothing$, then $e_\Delta$ cannot appear on the path from $p$ to $c$.\label{claim: D is valid}
\end{claim}
\begin{figure}[H]
	\centering
	\includegraphics[width=0.6\linewidth]{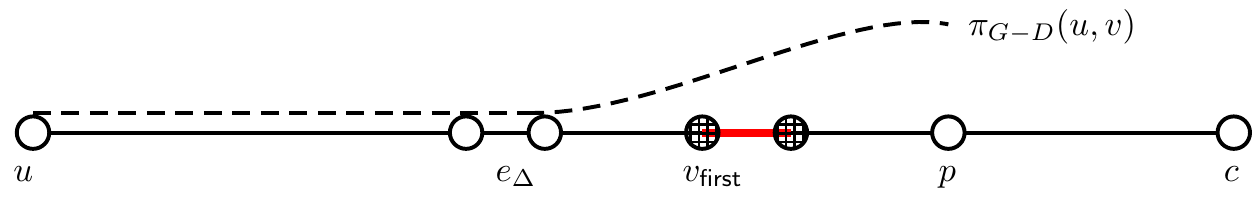}
\end{figure}
\begin{proof}
	Let $v_{\sf first}$ be the first vertex on $\pi(u, c)$ which is incident to a failed edge in $D$ on $\pi(u, c)$. Since $\pi_{G-D}(u, v)$ avoids the failed edge on $\pi(u, c)$, $e_\Delta$ has to be before $v_{\sf first}$. On the other hand, since $v_{\sf first} \in V(D)\subseteq\Key$, $v_{\sf first}$ does not lie after $p$. This means that $e_\Delta$ is strictly before $p$ on the path $\pi(u, c)$.
\end{proof}

If $D\cap \pi(u, c) = \varnothing$ then $D$ also satisfies (\ref{eq: condition for Dstar in Case II}). It is now valid to update
\[L \gets \min\{L, |\pi_{G - \Dstar}(u, v)|\}.\]
By the same reasoning as \cref{claim: warm-up}, if $|\pi_{G-D}(u, v)| < |\pi_{G-\Dstar}(u, v)|$, then $\pi_{G-D}(u, v)$ should go through some edge in $\Dstar$.

Now we construct a set $\Helper$ of candidate ``helper'' vertices $u'$ by examining every edge $e \in \Dstar \setminus D$ one by one. Suppose $e$ is an edge between $x$ and $y$:\begin{enumerate}[({Case } i)]

	\item If $\pi(x, v) \cap D = \varnothing$ or $\pi(y, v) \cap D = \varnothing$, we discard $e$.
	
	\includegraphics[width=0.8\linewidth]{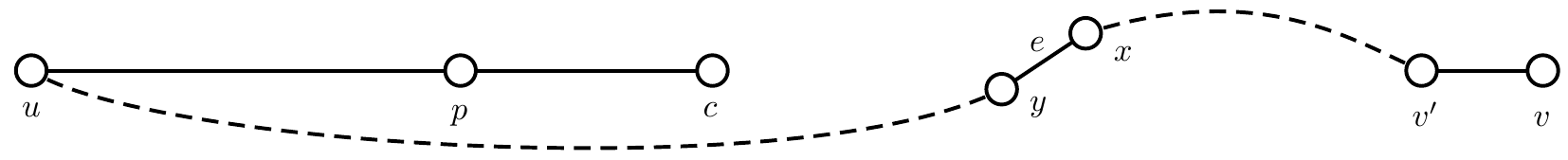}

	The reason is that $\pi_{G-D}(u, v)$ cannot go through $e$. Otherwise, suppose w.l.o.g.~that $\pi(x, v) \cap D = \varnothing$. Since $\pi_{G-D}(u, v)$ goes through $e$, in particular it also goes through $x$. It follows that $\pi(x, v)$ is a suffix of $\pi_{G-D}(u, v)$. Note that $\pi(v', v)$ is also a suffix of $\pi_{G-D}(u, v)$. Therefore $x$ is either an ancestor or a descendant of $v'$ in $T_v$.\begin{itemize}
		\item If $x$ is a descendant of $v'$ in $T_v$, then $x$ is both in $V(\Dstar)$ and $T_v(v')$, violating (\ref{eq: condition for Dstar in Case II}).
		\item If $x$ is an ancestor of $v'$ in $T_v$, then $e$ has to be on the path $\pi(v', v)$ in order for $\pi_{G-D}(u, v)$ to go through $e$, but this also violates (\ref{eq: condition for Dstar in Case II}).
	\end{itemize}

	\item Otherwise, if $\pi(u, x)\cap D \ne \varnothing$, we add $x$ into $H$. \label{item: case 2.2}
	
	Note that in this case, both $\pi(u, x)$ and $\pi(x, v)$ intersect $D$, therefore it is safe to add $x$ into $H$.
	
	\item Otherwise, if $\pi(u, y)\cap D \ne \varnothing$, we add $y$ into $H$.
	
	This case is similar as (Case \ref{item: case 2.2}).

	\item Otherwise, if $e$ is not a tree edge in $T_u$, we discard $e$.\label{item: case 2.4}

	\includegraphics[scale=0.8]{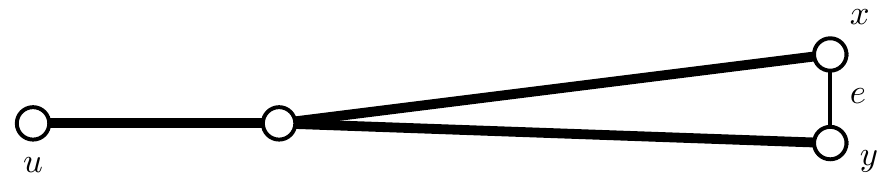}
	
	This is because in this case, both $\pi(u, x)$ and $\pi(u, y)$ are intact from $D$, and $\pi_{G-D}(u, v)$ does not need to go through $e$ at all. (For example, if $\pi_{G-D}(u, v)$ goes through $e$ and $x$ appears just before $y$, then $\pi_{G-D}(u, v)$ should use the path $\pi(u, y)$ instead of the concatenation of $\pi(u, x)$ and $e$.)
	
	\item Otherwise, $e$ is a tree edge in $T_u$. W.l.o.g.~we assume that $x$ is the parent of $y$ in $T_u$. If $T_u(y)\cap V(D) = \varnothing$, we add $y$ into $\Helper$.
	
	Note that $y$ is $u$-clean in this case, so it is safe to add it into $\Helper$.
	
	\item Otherwise, we discard $e$.\label{item: case 2.6}
	
	\includegraphics[width=0.7\linewidth]{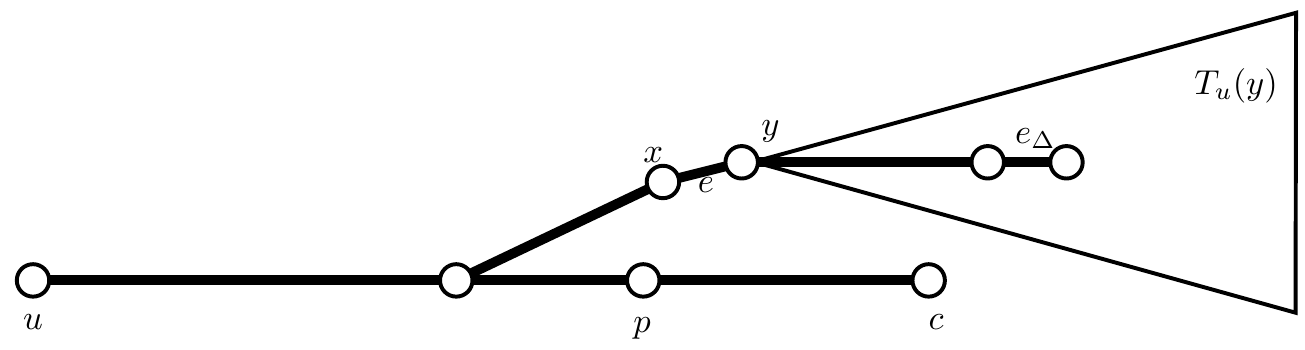}

	We need to prove that it is valid to discard $e$ in this case, i.e.~$\pi_{G-D}(u, v)$ cannot go through $e$. Note that since $T_u(y)\cap V(D) \ne \varnothing$, $y$ lies on $\Tinduced$ (i.e.~has non-zero degree in $\Tinduced$). Since $\pi(u, y) \cap D = \varnothing$, if $\pi_{G-D}(u, v)$ goes through $y$, then $\pi(u, y)$ must be a prefix of $\pi_{G-D}(u, v)$. It follows that $e_\Delta$ is either equal to $e$ or in the subtree $T_u(y)$. Since (\ref{eq: condition for Dstar in Case II}) implies that $e$ is not on the path $\pi(u, c)$, $e_\Delta$ cannot be on the path $\pi(p, c)$ either, a contradiction.
\end{enumerate}

\paragraph{Summary.} In Case II, we first construct the trees $\Tinduced$ and $\Tkey$ in $\tilde{O}(d)$ time. Then we enumerate an edge $(p, c) \in E(\Tkey)$. Let $D_\star = \caD[u, v, c, v', 0, 1]$, then for each edge $e\in D_\star\setminus D$, according to the above case analysis, we either discard $e$, add a vertex into $H$, or add a vertex into $\Helper$. Note that for every $u' \in \Helper$, $(D, u', v')$ satisfies (\ref{eq: requirement}).

After enumerating all edges in $E(\Tkey)$, we have that $|H| \le O(d^2)$ and $|\Helper| \le O(d^2)$. Then for each $u' \in \Helper$, we invoke the algorithm for Case I where we assume that $\pi_{G-D}(u, v)$ goes through both $u'$ and $v'$. Each invocation returns a hitting set of size $O(d)$ and an upper bound $L$ of $|\pi_{G-D}(u, v)|$. Finally, we let $L$ be the smallest upper bound found during the entire execution of the algorithm, and $H$ be the union of all hitting sets (which has size $O(d^3)$). It is easy to see that Case II takes $O(d^3)$ time.

\subsubsection{Case III}
Case III is the most general case: We only know the query $(u, v, D)$ but no ``helper'' vertices $u'$ or $v'$. The goal is to find a few intermediate vertices $w$ which are either $u$-clean or $v$-clean, such that $\pi_{G-D}(u, v)$ goes through one of the vertices $w$.

We construct the trees $\Tinduced^u$, $\Tkey^u$, $\Tinduced^v$, and $\Tkey^v$. Let $e_\Delta$ be the last edge on $\pi_{G-D}(u, v)$ such that the portion from $u$ to $e_\Delta$ on $\pi_{G-D}(u, v)$ is entirely in $\Tinduced^u$, and $e_\nabla$ be the first edge such that the portion from $e_\nabla$ to $v$ on $\pi_{G-D}(u, v)$ is entirely in $\Tinduced^v$.

We enumerate edges $(p^u, c^u) \in E(\Tkey^u)$ and $(c^v, p^v) \in E(\Tkey^v)$ such that $e_\Delta$ is on the path from $p^u$ to $c^u$ in $\Tinduced^u$, and $e_\nabla$ is on the path from $c^v$ to $p^v$ in $\Tinduced^v$. Note that there are $O(d^2)$ possible choices of $(p^u, c^u)$ and $(p^v, c^v)$.

\begin{figure}[H]
	\centering
	\includegraphics[scale=0.8]{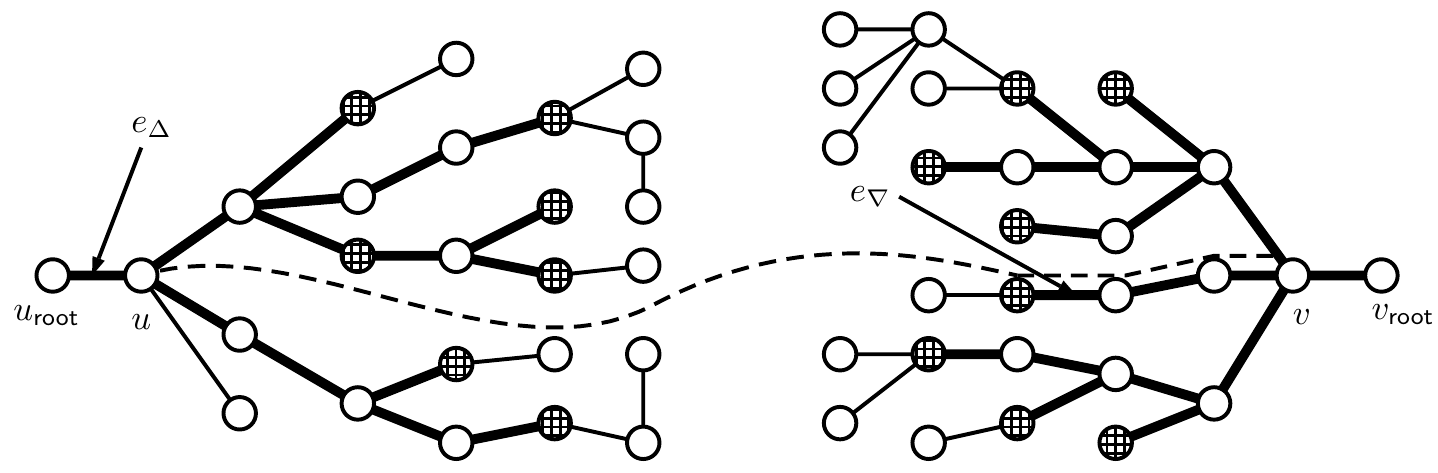}
	\caption{The trees $\Tinduced^u$ (left) and $\Tinduced^v$ (right). Again, the hatched vertices are vertices in $V(D)$, and the dashed curve corresponds to $\pi_{G-D}(u, v)$. Note that in this particular figure, $e_\Delta$ coincides with the edge between $\uroot$ and $u$.}
\end{figure}

\begin{figure}[H]
	\centering
	\includegraphics[scale=0.8]{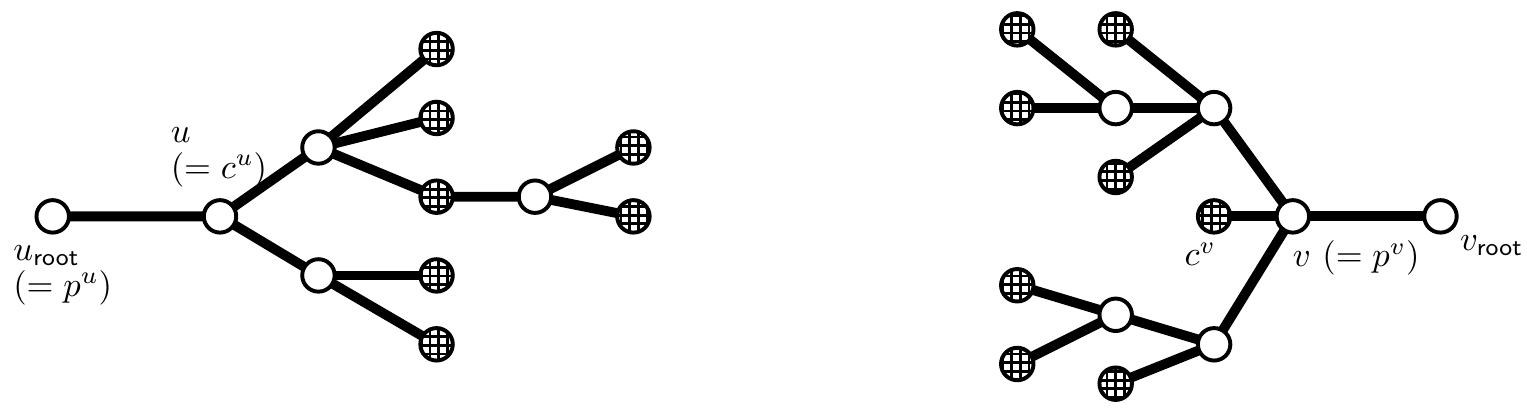}
	\caption{The trees $\Tkey^u$ (left) and $\Tkey^v$ (right).}
\end{figure}

Now let $\Dstar = \caD[u, v, c^u, c^v, 0, 0]$, then $\Dstar$ maximises $|\pi_{G-\Dstar}(u, v)|$ over all size-$d$ set of edge failures such that
\begin{equation}
\Dstar \cap \pi(u, c^u) = \varnothing, \text{ and }\Dstar \cap \pi(c^v, v) = \varnothing.\tag{$\gamma$}
\label{eq: condition for Dstar in Case III}
\end{equation}
By \cref{claim: D is valid}, if $D\cap \pi(u, c^u)\ne\varnothing$, then $e_\Delta$ cannot appear in the path from $p^u$ to $c^u$. Similarly, if $D\cap \pi(c^v, v)\ne\varnothing$, then $e_\nabla$ cannot appear in the path from $c^v$ to $p^v$. Therefore we may assume $D$ satisfies \cref{eq: condition for Dstar in Case III}, as otherwise we can discard $(p^u, c^v)$ and $(p^v, c^v)$. This means that it is safe to update
\[L \gets \min\{L, |\pi_{G - \Dstar}(u, v)|\}.\]

If $|\pi_{G-D}(u, v)| < |\pi_{G - \Dstar}(u, v)|$, then $\pi_{G-D}(u, v)$ should go through some edge in $\Dstar \setminus D$. Now we compute a hitting set $H$ and a set $\Helper$ of ``helper'' vertices by inspecting every edge in $\Dstar\setminus D$. In particular, we enumerate this edge $e = (x, y) \in \Dstar \setminus D$, and assume that $x$ appears before $y$ on the path $\pi_{G-D}(u, v)$. (That is, every edge $(x, y)$ is considered twice, once for $(x, y)$ and once for $(y, x)$.)
\begin{enumerate}[({Case }i)]
	\item Suppose that $\pi(u, x) \cap D = \varnothing$ and $\pi(y, v)\cap D = \varnothing$.
	
	We can immediately update $L \gets \min\{L, |\pi(u, x)| + w(e) + |\pi(y, v)|\}$. (Here $w(e)$ is the weight of the edge $e$.)

	\item Suppose that $\pi(u, x)\cap D \ne \varnothing$ and $\pi(y, v)\cap D \ne \varnothing$.
	
	If $\pi(x, v) \cap D = \varnothing$, then $y$ cannot appear on $\pi_{G-D}(x, v)$, which means the edge $e$ is invalid. Otherwise we can safely add $x$ into $H$.

	\item Suppose that $\pi(u, x)\cap D = \varnothing$ but $\pi(y, v)\cap D\ne \varnothing$. \label{item: case 3.3}
	
	\begin{enumerate}[({Case iii.}a)]
		\item Suppose that $(x, y)$ is not a tree edge in $T_u$. If $\pi(u, y)\cap D\ne\varnothing$, then we can safely add $y$ into $H$. Otherwise, a similar argument as (Case \ref{item: case 2.4}) in \cref{sec: case II} shows that we can discard $e$.
		\item Suppose that $(x, y)$ is a tree edge in $T_u$. Since $x$ appears before $y$ on $\pi_{G-D}(u, v)$, $x$ has to be the parent of $y$ in $T_u$ (otherwise we discard $(x, y)$). If $T_u(y)\cap V(D) = \varnothing$, we add $y$ into $\Helper$; otherwise we discard $y$.
		
		\includegraphics[width=0.8\linewidth]{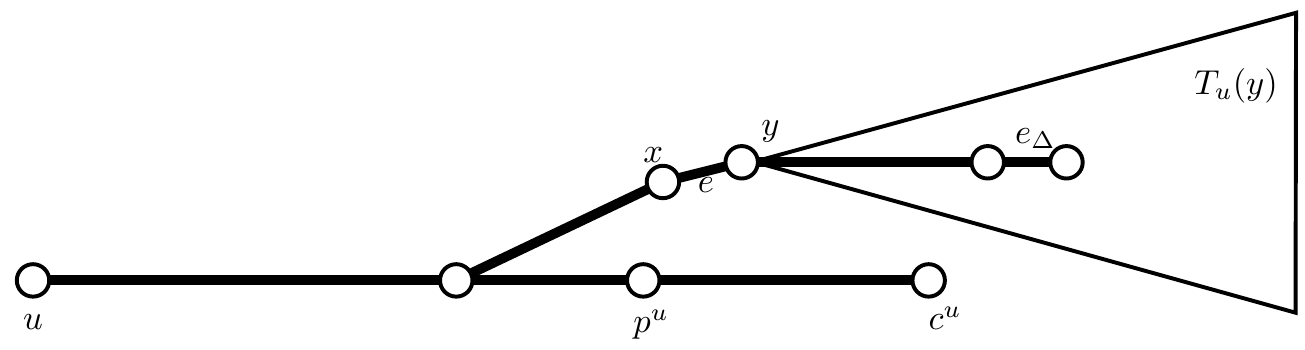}
		
		It is easy to see that if we add $y$ into $\Helper$, then $y$ is $u$-clean. Now we need to show that whenever we discard $y$, $\pi_{G-D}(u, v)$ cannot go through the edge $(x, y)$ (in the order of first $x$ and then $y$). This is essentially the same as (Case \ref{item: case 2.6}) in \cref{sec: case II}. Note that since $T_u(y) \cap V(D) \ne \varnothing$, $y$ lies on the tree $\Tinduced^u$. Since $\pi(u, y)\cap D = \varnothing$, if $\pi_{G-D}(u, v)$ goes through $y$, then $\pi(u, y)$ must be a prefix of $\pi_{G-D}(u, v)$. It follows that $e_\Delta$ is either equal to $e$ or in the subtree $T_u(y)$. By (\ref{eq: condition for Dstar in Case III}), $e$ is not on the path $\pi(u, c^u)$, therefore $e_\Delta$ cannot be on the path $\pi(p^u, c^u)$ either, a contradiction.
	\end{enumerate}
	
	\item Suppose that $\pi(u, x)\cap D \ne \varnothing$ but $\pi(y, v)\cap D = \varnothing$.
	
	This case is symmetric to (Case \ref{item: case 3.3}), so we only provide a sketch. If $(x, y)$ is not a tree edge in $T_v$, then we add $x$ into $H$ if $\pi(x, v) \cap D\ne\varnothing$ and discard $e$ otherwise. If $(x, y)$ is a tree edge in $T_v$, then (assuming $y$ is the parent of $x$ in $T_v$) we add $x$ into $\Helper$ if $T_v(x)\cap V(D) = \varnothing$ and discard $e$ otherwise.
\end{enumerate}

\paragraph{Summary.} In Case III, we first construct $\Tinduced^u$, $\Tinduced^v$, $\Tkey^u$, and $\Tkey^v$ in $\tilde{O}(d)$ time. Then we enumerate an edge $(p^u, c^u) \in \Tkey^u$ and an edge $(p^v, c^v) \in \Tkey^v$. Let $\Dstar = \caD[u, v, c^u, c^v, 0, 0]$, then for each edge $e \in \Dstar\setminus D$ and each of its two possible orientations, according to the above case analysis, we either discard $e$, add a vertex into $H$, or add a vertex into $\Helper$.

After this procedure, we have that $|H| \le O(d^3)$ and $|\Helper| \le O(d^3)$. In this way, we can reduce Case III to $O(d^3)$ instances of Case II. Note that an instance of Case II runs in $O(d^3)$ time, therefore an invocation of Case III runs in $O(d^6)$ time.

%% file: conclusion.tex
\section{Conclusions and Open Problems}
In this paper, we presented the first exact distance oracle that tolerates $d$ edge failures and has reasonable size and query time bounds. Our oracle has size $O(dn^4)$ and query time $d^{O(d)}$. However, our oracle still has some drawbacks:

\begin{enumerate}
	\item We think the biggest drawback of our oracle is its preprocessing time. Is there a faster preprocessing algorithm for our oracle? In particular, can we preprocess it in $O(n^c)$ time for some constant $c$ independent of $d$?
	\item Can we maintain exact distances under $d$ \emph{vertex} failures? Our oracle relies heavily on \cref{thm:edge-decomposable} which only works for edge failures.\footnote{\cite[Section 4.1]{DuanGR21} proved a version of \cref{thm:edge-decomposable} for \emph{approximate} distances under vertex failures. As the authors observe in their Footnote 15, a version for exact distances under vertex failures is unlikely to exist.}
	\item Can we improve the size of our oracle to (say) $\tilde{O}(dn^2)$? Currently our oracle is trivial when $d = 3$ or $d = 4$ and only non-trivial when $d > 4$. Such an improvement would imply non-trivial solutions for all $d$.
\end{enumerate}